\documentclass[11pt]{article}

\usepackage{amsthm}
\usepackage{graphicx} 
\usepackage{array} 

\usepackage{amsmath, amssymb, amsfonts, verbatim}
\usepackage{hyphenat,epsfig,subfigure,multirow}

\usepackage[usenames,dvipsnames]{xcolor}
\usepackage[ruled]{algorithm2e}

\usepackage{tcolorbox}
\tcbuselibrary{skins,breakable}
\tcbset{enhanced jigsaw}

\usepackage[compact]{titlesec}

\definecolor{DarkRed}{rgb}{0.5,0.1,0.1}
\definecolor{DarkBlue}{rgb}{0.1,0.1,0.5}

\usepackage{hyperref}
\hypersetup{
colorlinks=true,
pdfnewwindow=true,
citecolor=ForestGreen,
linkcolor=DarkRed,
filecolor=DarkRed,
urlcolor=DarkBlue
}

\usepackage{bm}
\usepackage{url}
\usepackage{xspace}
\usepackage[mathscr]{euscript}

\usepackage{mdframed}

\usepackage[noend]{algpseudocode}
\makeatletter
\def\BState{\State\hskip-\ALG@thistlm}
\makeatother

\usepackage{cite}
\usepackage{enumitem}

\usepackage[margin=1in]{geometry}

\newtheorem{theorem}{Theorem}
\newtheorem{lemma}{Lemma}[section]

\newtheorem{claim}[lemma]{Claim}
\newtheorem{fact}[lemma]{Fact}

\newtheorem{definition}{Definition}

\newtheorem{problem}{Problem}
\newtheorem{remark}[lemma]{Remark}

\newtheorem*{claim*}{Claim}
\newtheorem*{proposition*}{Proposition}
\newtheorem*{lemma*}{Lemma}
\newtheorem*{problem*}{Problem}

\newtheorem{mdresult}[theorem]{Theorem}
\newenvironment{Theorem}{\begin{mdframed}[backgroundcolor=lightgray!40,topline=false,rightline=false,leftline=false,bottomline=false,innertopmargin=2pt]\begin{mdresult}}{\end{mdresult}\end{mdframed}}

\newtheorem{mdinvariant}{Invariant}
\newenvironment{invariant}{\begin{mdframed}[hidealllines=false,backgroundcolor=gray!10,innertopmargin=0pt]\begin{mdinvariant}}{\end{mdinvariant}\end{mdframed}}

\allowdisplaybreaks

\renewcommand{\qed}{\nobreak \ifvmode \relax \else
      \ifdim\lastskip<1.5em \hskip-\lastskip
      \hskip1.5em plus0em minus0.5em \fi \nobreak
      \vrule height0.75em width0.5em depth0.25em\fi}

\newcommand{\toShrink}{-.20cm}
\newcommand{\toShrinkEnu}{-.2cm}

\newcommand{\card}[1]{\left\vert{#1}\right\vert}

\newcommand{\set}[1]{\ensuremath{\left\{ #1 \right\}}}

\newcommand{\alg}{\ensuremath{\mathcal{A}}\xspace}

\newcommand{\etal}{{\it et al.\,}}
\newcommand{\goodetal}{{\it et al.\/}}

\newcommand{\FG}{\ensuremath{\mathcal{G}}}

\newenvironment{tbox}{\begin{tcolorbox}[
		enlarge top by=5pt,
		enlarge bottom by=5pt,
		 breakable,
		 boxsep=0pt,
                  left=4pt,
                  right=4pt,
                  top=10pt,
                  arc=0pt,
                  boxrule=1pt,toprule=1pt,
                  colback=white
                  ]
	}
{\end{tcolorbox}}

\newcommand{\textbox}[2]{
{
\begin{tbox}
\textbf{#1}
{#2}
\end{tbox}
}
}

\newcommand{\CONGEST}{\ensuremath{\mathcal{CONGEST}}\xspace}
\newcommand{\mis}{\ensuremath{\mathcal{M}}}

\title{Fully Dynamic Maximal Independent Set\\ with Sublinear Update Time}
\author{Sepehr Assadi\thanks{{\small{\tt sassadi@cis.upenn.com}}\ Supported in part by NSF grant CCF-1617851. } \\ University of Pennsylvania
\and Krzysztof Onak\thanks{{\small{\tt konak@us.ibm.com}}} \\ IBM Research
\and Baruch Schieber\thanks {{\small{\tt sbar@us.ibm.com}}} \\ IBM Research
\and Shay Solomon\thanks{{\small{\tt solo.shay@gmail.com}}} \\ IBM Research
}

\date{}

\begin{document}
\maketitle

\thispagestyle{empty}
\begin{abstract}
A maximal independent set (MIS) can be maintained in an evolving $m$-edge graph
by simply recomputing it from scratch in $O(m)$ time after each update.
But can it be maintained in time sublinear in $m$ in fully dynamic graphs?

\smallskip

We answer this fundamental open question in the affirmative. We present a \emph{deterministic} algorithm with amortized update time $O(\min\{\Delta,m^{3/4}\})$, where $\Delta$ is a fixed bound on the maximum degree in the graph and $m$ is the (dynamically changing) number of edges.  

\smallskip

We further present a distributed implementation of our algorithm with $O(\min\{\Delta,m^{3/4}\})$ amortized message complexity, and $O(1)$ amortized round complexity and adjustment complexity (the number of vertices
that change their output after each update). This strengthens a similar result by Censor-Hillel, Haramaty, and Karnin (PODC'16) that required an assumption of a non-adaptive oblivious adversary. 

\end{abstract}

\setcounter{page}{0}

\newpage

\section{Introduction}\label{sec:intro}

\emph{Dynamic graph algorithms} constitute an active area of research in theoretical computer science.
Their objective is to maintain a solution to a combinatorial problem in an input graph---for example, a minimum spanning tree or maximal matching---under insertion and deletion of edges.
The research on dynamic graph algorithms addresses the natural question of whether one essentially needs to recompute the solution from scratch after every update.

This question has been asked over the years for a wide range of problems such as connectivity~\cite{HK99,HLT01}, minimum spanning tree~\cite{Frederickson85,EGIN97,HK97,HLT01,Wulff-N17},
maximal matching~\cite{IL93,BGS11,NS13,Sol16}, approximate matching and vertex cover~\cite{IL93,OR10,NS13,GP13,BS15,BHI15,BS16,BHN16,BHN17,BCH17}, shortest
paths~\cite{ES81,King99,DI04,Thorup05,RZ11,BR11,HKN14a,HKN14b,HKNS15,BC16,ACK17}, and graph coloring~\cite{BM17,BCKLRRV17,BCHN18}
(this is by no means a comprehensive summary of previous results).

Surprisingly however, almost no work has been done for the prominent problem of maintaining a \emph{maximal independent set (MIS)} in dynamic graphs. Indeed, the only
previous result for this problem that we are aware of is due to Censor-Hillel \goodetal~\cite{CHK16}, who developed a \emph{randomized} algorithm for this problem in \emph{distributed} dynamic networks and left the \emph{sequential} case (the main focus of this paper)
as a major open question. We note that implementing their distributed algorithm in the sequential setting requires $\Omega(\Delta)$\footnote{\label{footnote:potential} It is \emph{not} clear whether $O(\Delta)$ time
is also sufficient for this algorithm or not; see Section~6 of their paper.} update time in \emph{expectation}, where $\Delta$ is a fixed upper bound on the degree of vertices in the graph and
can be as large as $\Theta(m)$ in sparse graphs.

The maximal independent set problem is of fundamental importance in graph theory
with natural  connections to a plethora of other basic problems, such as vertex
cover, matching, vertex coloring, and edge coloring (in fact, all these problems
can be solved approximately by finding an MIS, see, e.g., the paper of
Linial~\cite{Linial87}).  As a result, this problem has been studied extensively
in different settings, in particular in parallel and distributed
algorithms~\cite{Cook83,KW85,ABI86,Luby86,Linial87,PS96,BEK14,BEPS16,Ghaffari16}.
(We refer the interested reader to the papers of Barenboim
\goodetal~\cite{BEPS16} and Ghaffari~\cite{Ghaffari16} for the story of this
problem in these settings and a comprehensive summary of previous work.)

In this paper, we concentrate on sequential algorithms for maintaining a maximal independent set in a dynamic graph. Our results are also applicable to the dynamic distributed setting
and improve upon the previous work of Censor-Hillel~\goodetal~\cite{CHK16}.

\subsection{Problem Statement and Our Results}\label{sec:results}

Recall that a maximal independent set (MIS) of an undirected graph, is a \emph{maximal} collection of vertices subject to the restriction that no pair of vertices in the collection are \emph{adjacent}. In the maximal independent set problem,
the goal is to compute an MIS of the input graph.

We study the \emph{fully dynamic} variant of the maximal independent set problem in which the goal is to maintain an MIS of a dynamic graph $G$, denoted by $\mis := \mis(G)$, subject to a sequence of edge
insertions and deletions. When an edge change occurs, the goal is to maintain $\mis$ in time significantly faster than simply recomputing it from scratch. Our main result is the following:
\vspace{-3pt}
\begin{Theorem}\label{thm:dynamic-m}
Starting from an empty graph on $n$ fixed vertices, a maximal independent set can be maintained {deterministically} over any sequence of edge insertions and deletions in $O({m^{3/4}})$ {amortized update time},
where $m$ denotes the dynamic number of edges.
\end{Theorem}
\vspace{-3pt}


As a warm-up to our main result in Theorem~\ref{thm:dynamic-m}, we also present an extremely simple \emph{deterministic} algorithm for maintaining an MIS with $O(\Delta)$ amortized update time, where $\Delta$ is a fixed
upper bound on the maximum degree of the graph. Our algorithms can be combined together to achieve a \emph{deterministic $O(\min\set{\Delta,m^{3/4}})$ amortized update time algorithm for maintaining an MIS in dynamic graphs}.
This constitutes the first improvement on the update time required for this problem in fully dynamic graphs over the na\"{\i}ve $O(m)$ bound for all possible values of $m$.  We now elaborate more on the details of our algorithm in Theorem~\ref{thm:dynamic-m}.


\paragraph{Deterministic Algorithm.} An important feature of our algorithm in Theorem~\ref{thm:dynamic-m} is that it is deterministic. The distinction between deterministic and randomized algorithms
is particularly important in the dynamic setting as almost all existing randomized algorithms require the assumption of a \emph{non-adaptive oblivious adversary} who is not allowed to learn anything about the algorithm's random bits. Alternately, this setting can be
viewed as the requirement that the entire sequence of updates be fixed in advance, in which case the adversary cannot use the solution maintained by the algorithm in order to break its guarantees. While these assumptions can be naturally justified in many
settings, they can render randomized algorithms entirely unusable in certain scenarios (see, e.g.,~\cite{BC16,BS15,BHN16} for more details).

As a result of this assumption, obtaining a deterministic algorithm for most dynamic problems is considered a distinctively harder task compared to finding a randomized one.
This is evident by the polynomial gap between the update time of best known deterministic algorithms compared to randomized ones for many dynamic problems.
For example, a maximal matching can be maintained in a fully dynamic graph with $O(1)$ update
time via a randomized algorithm~\cite{Sol16}, assuming a non-adaptive oblivious adversary, while the best known deterministic algorithm for this problem requires $\Theta(\sqrt{m})$ update time~\cite{NS13} (see~\cite{BCHN18} for a similar
situation for $(\Delta+1)$-coloring of vertices of a graph).

\paragraph{${O(1)}$-Amortized Adjustment Complexity.} An important performance measure of a dynamic algorithm is its \emph{adjustment complexity} (sometimes called \emph{recourse}) that counts the number of vertices (or edges) that
need to be inserted to or deleted from the maintained solution after each update (see,~e.g.~\cite{CHK16,BCKLRRV17,GKKP17,BCHN18}).
For many natural graph problems such as maintaining a maximal matching, constant worst-case adjustment complexity
can be trivially achieved since one edge update cannot ever necessitate more than a constant number of changes in the maintained solution. This is, however, \emph{not} the case for the MIS problem: by inserting an edge
between two vertices already in $\mis$, the adversary can force the algorithm to delete at least one end point of this edge from $\mis$, which in turn forces the algorithm to pick \emph{all} neighbors of this deleted
vertex to ensure maximality (this phenomena also highlights a major challenge in the treatment of this problem compared to the maximal matching problem which we discuss further below). 

Nevertheless, we prove that the adjustment complexity of our algorithm in Theorem~\ref{thm:dynamic-m} is $O(1)$ on average which is clearly optimal. Can we further improve
our results to achieve an $O(1)$ \emph{worst-case} adjustment complexity? We claim that this is indeed not possible by showing that the worst-case adjustment complexity of any algorithm for maintaining an MIS is $\Omega(n)$, using
a simple adaption of an example proposed originally by~\cite{CHK16} for proving a similar result in distributed settings when vertex deletions are also allowed by the adversary (this follows seamlessly from the result in~\cite{CHK16} and is provided in
Appendix~\ref{app:worst-case-example} only for completeness).

\paragraph{Distributed Implementation.} Finding a maximal independent set is one of the most studied problems in distributed computing. In the distributed computing model, there is a processor on each vertex of the graph.
Computation proceeds in synchronous rounds during which every processor can communicate messages of size $O(\log n)$ with its neighbors (this corresponds to the
\CONGEST~model of distributed computation; see Section~\ref{app:distributed} for further details). In the dynamic setting, both edges and vertices can be inserted to or deleted from the graph and the goal is
to update the solution in a small number of rounds of communication, with small communication cost and adjustment complexity.

Our results in the sequential setting also imply a \emph{deterministic distributed algorithm for maintaining an MIS in a dynamic network
with $O(1)$ amortized round complexity, $O(1)$ amortized adjustment complexity, and $O(\min\set{\Delta,m^{3/4}})$ amortized message complexity per each update.}
This result achieves an improved message complexity compared to the distributed algorithm of~\cite{CHK16} with asymptotically the same round and adjustment complexity (albeit in amortized sense as opposed
to in expectation; see Section~\ref{app:distributed}).
More importantly, our result is achieved
via a \emph{deterministic} algorithm and does not require the assumption of a non-adaptive oblivious adversary.
Similar to~\cite{CHK16}, our algorithm can also be implemented in the asynchronous model, where there is no global synchronization of communication between nodes.
We elaborate more on this result in Section~\ref{app:distributed}.

\paragraph{Maximal Independent Set vs.\ Maximal Matching.}
We conclude this section by comparing the maximal independent set problem to the closely related problem of maintaining a maximal matching\footnote{A maximal matching in a graph $G$ can be obtained by computing an
MIS of the line graph of $G$.} in dynamic graphs. We discuss additional challenges that one encounters for the maximal independent set problem.

In sharp contrast to the maximal independent set problem, maintaining maximal matchings in dynamic graphs  has been studied extensively, culminating in an $O(\sqrt{m})$ worst-case update time  deterministic algorithms~\cite{NS13} and $O(1)$ expected update time randomized algorithm~\cite{Sol16} (assuming a non-adaptive oblivious adversary).

Maintaining an MIS in a dynamic graph seems inherently more complicated than maintaining a maximal matching. One simple reason is that as argued before, a single update can only change the status of $O(1)$ edges/vertices in the
maximal matching, while any algorithm can be forced to make $\Omega(n)$ changes to the MIS for a single edge update in the worst case. As a result, a maximal matching can be maintained with an $O(\Delta)$ worst-case update time
via a straightforward algorithm (see,~e.g.~\cite{IL93,OR10}), while the analogous approach for MIS only results in $O(m)$ update time.

Another, perhaps more fundamental difference between the two problems lies in their different level of ``locality.'' To adjust a maximal matching after
an update, we only need to consider the neighbors of the recently unmatched vertices (to find another unmatched vertex to match with), while to fix an MIS, we need to consider the two-hop neighborhood of a recently
removed vertex from the MIS (to add to the MIS the neighbors of this vertex which themselves do not have another neighbor in the MIS). We note that this difficulty is similar-in-spirit to the barrier for maintaining
a \emph{better than $2$} approximate matching via \emph{excluding length-$3$ augmenting paths} in dynamic graphs. Currently, the best known algorithm for achieving a better than $2$-approximation to matching in dynamic graphs
requires $O(m^{1/4})$ update time~\cite{BS15,BS16}. Achieving sub-polynomial in $m$ update time---even using randomness and assuming a non-adaptive oblivious adversary---remains
a major open problem in this area (we refer the interested reader to~\cite{BHN16} for more details).

We emphasize that even for the seemingly easier problem of maximal matching, the best upper bound on update time using a deterministic algorithm (the focus of our paper) is only $O(\sqrt{m})$~\cite{NS13}.

\subsection{Overview of Our Techniques}\label{sec:techniques}

\paragraph{$O(\Delta)$-Amortized Update Time.} Consider the following simple algorithm for maintaining an MIS $\mis$ of a dynamic graph: for each vertex, maintain the number of its neighbors in $\mis$ in a counter, and
for each update in the graph or $\mis$, spend $O(\Delta)$ time to update this counter for the neighbors of the
updated vertex. What is the complexity of this algorithm? Unfortunately, as argued before, an update to the graph may inevitably result in an update of size $\Omega(n)$ to $\mis$. Processing it may take $\Omega(n \cdot \Delta)$ time
as we have to update all neighbors of every updated vertex. However, all we need to handle this case is the following basic observation: while a single update can force the algorithm
to insert up to $\Omega(n)$ vertices to $\mis$, it can never force the algorithm to remove more than one vertex from $\mis$. We therefore charge the $O(\Delta)$ time needed to insert
a vertex into $\mis$ (and there can be many such vertices per one update) to the time spent in a previous update in which the same vertex was (the only one) removed from $\mis$. This allows us to argue that on \emph{average}, we
only spend $O(\Delta)$ time per update.

\paragraph{$O(m^{3/4})$-Amortized Update Time.}
Achieving an $o(\Delta)$ amortized update time however is distinctly more challenging. On the one hand, we cannot afford to update all neighbors of a vertex after every change in the graph.
On the other hand, we do not have enough time to iterate over all neighbors of an updated vertex to even check whether or not they should be added to $\mis$ and hence need to
maintain this information, which is a function of vertices in the two-hop neighborhood of a vertex, explicitly for every vertex.

To bypass these challenges, we relax our requirement for
knowing the status of \emph{all} vertices in the neighborhood of a vertex, and instead maintain the status of some vertices that are in the two-hop neighborhood of a vertex. More concretely,
we allow ``high'' degree vertices to \emph{not} update their ``low'' degree neighbors about their status (as the number of low degree neighbors can be very large), while making every ``low'' degree vertex update not only all its neighbors but even some of its
neighbor's neighbors, using the extra time available to this vertex (as its degree is small). This approach allows us to maintain
a ``noisy'' version of the information described above. Note that this information is not completely accurate as the status of some vertices in $\mis$ would be unknown to their neighbors and their neighbor's neighbors (%
in the actual algorithm, we use a more fine-grained partition of vertices based on their degree into more than two classes, not only ``high'' and ``low'').

We now need to address a new challenge introduced by working with this ``noisy'' information:  we may decide that a vertex is ready to join $\mis$ based on the information stored in the algorithm and insert this vertex to $\mis$, only to find out
that there are already some vertices in $\mis$ adjacent to this vertex. To handle this, we also relax the property of the basic algorithm above that only allowed for deleting one vertex from $\mis$ per each update.
This allows us to insert multiple vertices to $\mis$ as long as a large portion (but not all) of their neighbors are known to be not in $\mis$. Then we go back and delete a small number of ``violating'' vertices
from $\mis$ to make sure it is indeed an independent set. Note that deleting those vertices may now require inserting a new set vertices in their neighborhood to $\mis$ to ensure maximality.

In order to be able to perform all those operations and recursively treat the newly deleted vertices in a timely manner, we maintain the invariant that whenever we need to remove more than one vertex from $\mis$, the number of inserted vertices
leading to this case is much larger than the number of removed vertices. This allows us to extend the simple charging scheme used in the analysis of the basic algorithm above to this new algorithm and prove our upper bound on the amortized update time
of the algorithm.

We point out that despite the multiple challenges along the way that are described above, our algorithm turned to be quite simple in hindsight. The main delicate matters are in the choice of parameters and in the analysis. This in turn makes the implementation of our results in sequential and distributed settings quite practical.

\paragraph{Organization.} We introduce our notation and preliminaries in Section~\ref{sec:prelim}. We then present a simple proof of the $O(\Delta)$-amortized update time algorithm in Section~\ref{app:dynamic-delta} as a warm-up to
our main result. Section~\ref{sec:dynamic-m} contains the proof of our main result in Theorem~\ref{thm:dynamic-m}. The distributed implementation of our result
and a detailed comparison of our results with that of Censor-Hillel~\goodetal~\cite{CHK16} appear in Section~\ref{app:distributed}.


\section{Preliminaries}\label{sec:prelim}

\paragraph{Notation.} We denote the static vertex set of the input graph by $V$. Let $\FG = \langle{G_0,G_1,\ldots \rangle}$ be the sequence of graphs that are given to the algorithm: 
initial graph $G_0$ is empty and each graph $G_t$ is obtained from the previous graph $G_{t-1}$ by either inserting or deleting a single edge $e_t = (u_t,v_t)$. We use $G_t(V,E_t)$ to denote
the graph at step $t$ and define $m_t := \card{E_t}$. Finally, throughout the paper, $\mis$ denotes the maximal independent set maintained by the algorithm at every time step. 


\paragraph{Greedy MIS Algorithm.} Consider the following algorithm for computing an MIS of a given graph: Fix an arbitrary ordering of the vertices in the graph, add the first vertex to the MIS, remove all its neighbors from the list, and continue. 
This algorithm clearly computes an MIS of the input graph. In the rest of the paper, we refer to this algorithm as \emph{the greedy MIS algorithm}. 
\vspace{-5pt}
\begin{fact}\label{fact:greedy-mis}
	For an $n$-vertex graph $G$ with maximum degree $\Delta$, the greedy MIS algorithm computes an MIS of size at least $n/(\Delta+1)$. 
\end{fact}

\newcommand{\MISList}[1]{\ensuremath{\textnormal{\texttt{LISSST}[$#1$]}}\xspace}
\newcommand{\MISCounter}[1]{\ensuremath{\textnormal{\texttt{MISCounter}[$#1$]}}\xspace}

\section{Warm-Up: A Simple $O(\Delta)$-Update-Time Dynamic Algorithm}\label{app:dynamic-delta}

As a warm-up to our main result, we describe a straightforward algorithm for maintaining an MIS $\mis$ in a dynamic graph with
$O(\Delta)$ amortized update time, where $\Delta$ is a fixed upper bound on the maximum degree in the graph.  For every vertex $v$ in the graph, we simply maintain a counter $\MISCounter{v}$, counting number of its neighbors in $\mis$.
In the following, we consider updating $\mis$ and this counter after each edge update.

Let $e_t=(u_t,v_t)$ be the updated edge. Suppose first that we delete this edge. In this case, $u_t$ and $v_t$ cannot both be in $\mis$ by definition of an independent set. Also, if none of them belong to $\mis$, there is nothing to do.
The interesting case is thus when exactly one of $u_t$ or $v_t$ belongs to $\mis$; without loss of generality, we assume this vertex is $u_t$. We first subtract one from $\MISCounter{v_t}$ (as it is no longer adjacent to $u_t$). If $\MISCounter{v_t} > 0$ still,
it means that $v_t$ is adjacent to some vertex in $\mis$ and hence we are done.  Otherwise, we add $v_t$ to $\mis$ and update the counter of all its neighbors in $O(\Delta)$ time.
Clearly, this step takes $O(\Delta)$ time in the worst case, after that $\mis$ is indeed an MIS.

Now suppose $e_t$ was inserted to the graph. The only interesting case here is when both $u_t$ and $v_t$ belong to $\mis$ (we do not need to do anything in the remaining cases, other than perhaps updating the neighbor list of $u_t$ and $v_t$ in $O(1)$ time).
To ensure that $\mis$ remains an independent set, we need to remove one of these vertices, say $u_t$, from $\mis$. After this, to ensure the maximality, we have to insert to $\mis$ any
neighbor of $u_t$ that can now join $\mis$. To do this, we first update the $\MISCounter{\cdot}$ of all neighbors of $u_t$ in $O(\Delta)$ time. Next (using the updated counter), we iterate over all neighbors of $u_t$ and for each one check if they can be inserted to $\mis$ now or not. 
If so, we add this new vertex to $\mis$ and inform all its neighbors in $O(\Delta)$ time to update their $\MISCounter{\cdot}$.
It is easy to see that in this case, we spend $O(k \cdot \Delta)$ time in the worst case, where $k$ is the number of vertices added to $\mis$. 

The correctness of this algorithm is straightforward to verify. We now prove that the amortized running time of the algorithm is $O(\Delta)$. The crucial observation is that whenever we change $\mis$, we may increase its size without any restriction,
but we never decrease its size by more than one. We use the following straightforward charging scheme.

Initially, we start with all vertices being in $\mis$ as the original graph is empty. Whenever we delete one vertex from $\mis$, we spend $O(\Delta)$ time to handle this vertex (including updating its neighbors and checking which ones can
join $\mis$), and place $O(\Delta)$ ``extra budget'' on this vertex to be spent later. We use this budget when this vertex is being inserted to $\mis$ again. Whenever we want to bring this vertex back to $\mis$, we only need to spend
this extra budget and hence the $O(\Delta)$ time spent for inserting this vertex back to $\mis$ can be charged to the time spent for this vertex when we removed it from $\mis$.
This implies that the update time is $O(\Delta)$ in average. We can therefore conclude the following lemma.

\begin{lemma}\label{lem:dynamic-delta}
	Starting from an empty graph on $n$ vertices, a maximal independent set can be maintained \emph{deterministically} over any sequence of $K$ edge insertions and deletions in $O(K \cdot \Delta)$ time
	where $\Delta$ is a fixed bound on the maximum degree in the graph.
\end{lemma}

\newcommand{\Deg}[1]{\ensuremath{\textnormal{\texttt{degree}[$#1$]}}\xspace}
\newcommand{\Nei}[1]{\ensuremath{\textnormal{\texttt{neighbors}[$#1$]}}\xspace}
\newcommand{\NeiDeg}[1]{\ensuremath{\textnormal{\texttt{neighbors-degree}[$#1$]}}\xspace}
\newcommand{\MISFlag}[1]{\ensuremath{\textnormal{\texttt{MIS-flag}[$#1$]}}\xspace}
\newcommand{\MISNei}[1]{\ensuremath{\textnormal{\texttt{MIS-neighbors}[$#1$]}}\xspace}
\newcommand{\MIStwoHopNei}[1]{\ensuremath{\textnormal{\texttt{MIS-2hop-neighbors}[$#1$]}}\xspace}

\newcommand{\vmis}{\ensuremath{V_{\textnormal{\textsf{MIS}}}}}

\newcommand{\vh}{\ensuremath{V_{\textnormal{\textsf{High}}}}}
\newcommand{\vmh}{\ensuremath{V_{\textnormal{\textsf{Med-High}}}}}
\newcommand{\vml}{\ensuremath{V_{\textnormal{\textsf{Med-Low}}}}}
\newcommand{\vl}{\ensuremath{V_{\textnormal{\textsf{Low}}}}}

\newcommand{\UpdateNeighbors}{\ensuremath{\textnormal{\textsf{UpdateNeighbors}}}\xspace}
\newcommand{\UpdateTwoHopNeighbors}{\ensuremath{\textnormal{\textsf{UpdateTwoHopNeighbors}}}\xspace}

\newcommand{\UpdateCrossing}{\ensuremath{\textnormal{\textsf{UpdateCrossingCase}}}\xspace}
\newcommand{\UpdateMISInsertion}{\ensuremath{\textnormal{\textsf{UpdateMISInsertionCase}}}\xspace}

\newcommand{\Lmis}{\ensuremath{L_{\textnormal{\textsf{MIS}}}}}
\newcommand{\Lonehop}{\ensuremath{L_{\textnormal{\textsf{1-hop}}}}}
\newcommand{\Ltwohop}{\ensuremath{L_{\textnormal{\textsf{2-hop}}}}}

\newcommand{\ellmis}{\ensuremath{\ell_{\textnormal{\textsf{MIS}}}}}
\newcommand{\ellonehop}{\ensuremath{\ell_{\textnormal{\textsf{1-hop}}}}}
\newcommand{\elltwohop}{\ensuremath{\ell_{\textnormal{\textsf{2-hop}}}}}

\newcommand{\tk}{\ensuremath{\widetilde{k}}}

\section{An $O(m^{3/4})$-Update-Time Dynamic Algorithm}\label{sec:dynamic-m}

We present our fully dynamic algorithm for maintaining a maximal independent set in this section and prove Theorem~\ref{thm:dynamic-m}.
The following lemma is a somewhat weaker looking version of Theorem~\ref{thm:dynamic-m}. However, we prove next that this lemma is all we need to prove Theorem~\ref{thm:dynamic-m}.

\begin{lemma}\label{lem:dynamic-m}
	Starting with any arbitrary graph on $n$ vertices and $m$ edges, a maximal independent set $\mis$ can be maintained \emph{deterministically} over any sequence of $K = \Omega(m)$ edge insertions and deletions
	in $O(K \cdot m^{3/4})$ time, as long as the number of edges
	remains within a factor $2$ of $m$.
\end{lemma}
We first show that this lemma implies Theorem~\ref{thm:dynamic-m}.
\begin{proof}[Proof of Theorem~\ref{thm:dynamic-m}]
	For simplicity, we define $m=1$ in case of empty graphs. We start from the empty graph and run the algorithm in Lemma~\ref{lem:dynamic-m} until the number of edges $m_t$ in the graph differs from $m$
	by a factor more than $2$. This crucially implies that the total number of updates before terminating the algorithm (the parameter $K$ in Lemma~\ref{lem:dynamic-m}), is $\Omega(m)$. As such, we can invoke
	Lemma~\ref{lem:dynamic-m} to obtain an upper bound of $O(m^{3/4})$ on the amortized update time of the algorithm throughout these updates. We then update $m= m_t$
	and start running the algorithm in Lemma~\ref{lem:dynamic-m} on the current graph using the new choice of $m$. Clearly, this results in an amortized update time of $O(m^{3/4})$ where $m$ now denotes the
	number of dynamic edges in the graph. As the algorithm in Lemma~\ref{lem:dynamic-m} always maintain an MIS of the underlying graph, we obtain the final result.
\end{proof}

The rest of this section is devoted to the proof of Lemma~\ref{lem:dynamic-m}. In the following, we first describe the data structure maintained in the algorithm for storing the required information and its main properties and then present our update algorithm.

\subsection{The Data Structure}\label{sec:dynamic-m-ds}


For every vertex $v$, we maintain the following information:
\begin{tbox}
\begin{itemize}[itemsep=0pt]
	\item $\Nei{v}$: a list of current neighbors of $v$ in the graph.
	\item $\Deg{v}$: an estimate of the degree of $v$ to within a factor of two.
	\item $\NeiDeg{v}$: a list containing $\Deg{u}$ for every vertex $u$ in $\Nei{v}$.
	\item $\MISFlag{v}$: a boolean entry indicating whether or not $v$ belongs to $\mis$.
	\item $\MISNei{v}$: a counter denoting the size of a suitable subset of current neighbors of $v$ in $\mis$. Any vertex counted in $\MISNei{v}$ belongs to $\mis$ but not all neighbors of $v$ in $\mis$ are (necessarily) counted in $\MISNei{v}$
	(see Invariant~\ref{inv:neighbors} for more detail).
	\item $\MIStwoHopNei{v}$: a list, containing for every vertex $w$ in $\Nei{v}$, a counter that counts the size of a suitable subset of current neighbors of $w$ in $\mis$.
	Any vertex counted in $\MIStwoHopNei{v}[w]$ is also counted in $\MISNei{w}$ but \emph{not} vice versa
	(see Invariant~\ref{inv:2hop-neighbors} below for more detail).
\end{itemize}
\end{tbox}
Additionally, we maintain a partition $(\vh,\vmh,\vml,\vl)$ of the vertices into four sets based on their current approximate degree, namely, $\Deg{v}$.
In particular, $v$ belongs to $\vh$ iff $\Deg{v} \geq m^{3/4}$, to $\vmh$ iff $m^{3/4} > \Deg{v} \geq m^{1/2}$, to $\vml$ iff $m^{1/2} > \Deg{v} \geq m^{1/4}$,
and to $\vl$ iff $\Deg{v} < m^{1/4}$. We refer to the vertices of $\vl$ as the \emph{low-degree} vertices.
Throughout, we assume that in any of the lists maintained for a vertex by the algorithm, we can directly iterate over vertices of a particular subset in $(\vh,\vmh,\vml,\vl)$.
(This can be done, for example, by storing these lists as four separate linked lists, one per each such subset.)

The following invariant is concerned with the information we need from $\MISNei{v}$.

\begin{invariant}\label{inv:neighbors}
	For any vertex $v \in V \setminus \vl$, $\MISNei{v}$ counts the number of all neighbors of $v$  in $\mis$.
	For any vertex $v \in \vl$, $\MISNei{v}$ counts the number of neighbors of $v$ that are in $\mis$ but \emph{not} in $\vh$, i.e., are in $\mis \setminus \vh$.
\end{invariant}

By Invariant~\ref{inv:neighbors}, any vertex either knows the number of \emph{all} its neighbors in $\mis$ or is a low-degree vertex
and can iterate over all its neighbors in $O(m^{1/4})$ time to count this number. Moreover, even a low-degree vertex knows
the number of its neighbors in $\mis \setminus \vh$. This is crucial for our algorithm as in some cases, we need to iterate over \emph{many} vertices that belong to $\vl$ and
decide if they can join $\mis$ and hence cannot spend $O(m^{1/4})$ time per each vertex to determine this information. Note that the information we obtain in this way is ``noisy'', as we ignore some neighbors of vertices in $\vl$ that
are potentially in $\mis$. We shall address this problem
using a post-processing step that exploits the fact that the total number of ignored vertices, i.e., vertices in $\vh$, is small.

The following invariant is concerned with the information we need from $\MIStwoHopNei{v}$.
\vspace{-10pt}
\begin{invariant}\label{inv:2hop-neighbors}
	For any $v \in V$ and and every $u \in \Nei{v} \cap \vl$, $\MIStwoHopNei{v}[u]$ counts the number of vertices in $\vml \cup \vl$ that belong to $\mis$ and are neighbors of $u$ (the entry in $\MIStwoHopNei{v}[u]$ for any vertex $u \notin \vl$ is
	$\bot$).
\end{invariant}

Invariant~\ref{inv:2hop-neighbors} allows us to infer some nontrivial information about the two-hop neighborhood of any vertex. We use Invariant~\ref{inv:2hop-neighbors}
to quickly determine which neighbors of a vertex $v$ can be added to $\mis$ in case $v$ is deleted from it. Similar to the one-hop information we obtain through maintaining Invariant~\ref{inv:neighbors}, the information we obtain in this way
 is also ``noisy''.

We show how to update the information per each vertex after a change in the topology or $\mis$.
Maintaining $\Nei{v}$ under edge updates is straightforward. To maintain $\Deg{v}$, each vertex simply keeps a $2$-approximation of its degree in $\Deg{v}$.
Whenever the current actual degree of $v$ differs from $\Deg{v}$ by more than a factor of two, $v$ updates $\Deg{v}$ to its actual degree and informs {all} its neighbors $u \in \Nei{v}$ to update $\NeiDeg{u}$.
This requires only $O(1)$ amortized time.

The above information is a function of the underlying graph and not $\mis$. We also need to update the information per each vertex that are functions of $\mis$
whenever $\mis$ changes. Maintaining $\MISFlag{v}$ is trivial for any vertex $v$, hence in the following we focus on the remaining two parts.

Once a vertex $u$ changes its status in $\mis$, we apply the following algorithm to
update the value of $\MISNei{v}$ for every vertex $v$  (we only need to update this for $v \in \Nei{u}$).

\textbox{Algorithm $\UpdateNeighbors(u)$. \textnormal{An algorithm called whenever a vertex $u$ enters or exists $\mis$ to update $\MISNei{v}$ for neighbors of $u$.}}{
\begin{enumerate}[itemsep=0pt]
	\item If $u \in \vh$, update $\MISNei{v}$ for any vertex $v \in \Nei{u}$ \emph{not} in $\vl$ accordingly (i.e., add or subtract one depending on whether $u$ joined or left $\mis$).
	\item If $u \notin \vh$, update $\MISNei{v}$ for every vertex $v \in \Nei{u}$.
\end{enumerate}
}

It is immediate to see that by running $\UpdateNeighbors(u)$ in our main algorithm whenever a vertex $u$ is updated in $\mis$, we can maintain Invariant~\ref{inv:neighbors}. Also each call to $\UpdateNeighbors(u)$
takes $O(m^{3/4})$ time in worst-case since in both cases of the algorithm, we only need to update $O(m^{3/4})$ vertices: $(i)$ if $u \in \vh$, the algorithm only updates the vertices in $V \setminus \vl$ whose size is
$O(m^{3/4})$, and $(ii)$ if $u \notin \vh$, $u$ only has $O(m^{3/4})$ neighbors to update. We also point out that $\MISNei{v}$ can be updated easily whenever an edge incident on $(u,v)$ is inserted or deleted in $O(1)$ time
by simply visiting $\MISFlag{u}$ and updating $\MISNei{v}$ accordingly.

Now consider updating $\MIStwoHopNei{\cdot}$. We use the following algorithm on a vertex $u$ that has changed its status in $\mis$
to update $\MIStwoHopNei{v}$ for every vertex $v$ in the
graph (we only need to update this information for the two-hop neighborhood of $u$).

\textbox{Algorithm $\UpdateTwoHopNeighbors(u)$. \textnormal{An algorithm called when a vertex $u$ enters or exists $\mis$ to update $\MIStwoHopNei{v}$ for the two-hop neighborhood of $u$. }}{
\begin{enumerate}[itemsep=0pt]
	\item If $u \in \vml \cup \vl$, for any vertex $w \in \Nei{u}$:
	\begin{enumerate}[itemsep=0pt]
		\item If $w$ belongs to $\vl$, iterate over all vertices $v \in \Nei{w}$.
		\item For any such $v$, update $\MIStwoHopNei{v}[w]$ accordingly (i.e., add or subtract one depending on whether $u$ joined or left $\mis$).
	\end{enumerate}
\end{enumerate}
}

Each call to $\UpdateTwoHopNeighbors(u)$ takes $O(m^{3/4})$ time in the worst case. This is because $u$ only updates its neighbors if it has $O(m^{1/2})$ neighbors as $u$ should be in $\vml \cup \vl$ and
when it updates its neighbors, it changes the counter of $O(m^{1/4})$ vertices (as $w$ should be in $\vl$). This ensures that the running time of the algorithm is $O(m^{1/2} \cdot m^{1/4}) = O(m^{3/4})$.
Whenever an edge $(u,v)$ is updated in the graph, we can run a similar algorithm to update the two-hop neighborhood of $u$ and $v$ in the same way in $O(m^{3/4})$ time; we omit the details.
It is also straightforward to verify that by running $\UpdateTwoHopNeighbors(u)$ in our main algorithm whenever a vertex $u$ is updated in $\mis$, we preserve Invariant~\ref{inv:2hop-neighbors}.

Finally, recall that we also need a preprocessing step that given a graph $G$ initializes this data structure. We can implement this step by first initializing all non-MIS-related information in this data structure in $O(m)$ time (we do not need to
handle isolated vertices at this point).
Next, we run the greedy MIS algorithm to compute an MIS  $\mis$ of this graph in $O(m)$ time (again only on non-isolated vertices).
Finally, we update the information for every vertex in $\mis$ using the two procedures above which takes $O(m \cdot m^{3/4})$ time in total.
(We note that a more efficient implementation for this initial stage is possible.)
As $K = \Omega(m)$ in Lemma~\ref{lem:dynamic-m}, this (one time only) initialization cost is within the bounds stated in the lemma statement.

We summarize the results in this section in the following two lemmas.
\begin{lemma}\label{lem:ds}
	After updating any single edge, the data structure stored in the algorithm can be updated in $O(m^{3/4})$ \emph{amortized time}. Moreover, Invariants~\ref{inv:neighbors} and~\ref{inv:2hop-neighbors}
	hold after this update.
\end{lemma}

\begin{lemma}\label{lem:mis-update}
	After updating any single vertex in $\mis$, the data structure stored in the algorithm can be updated in $O(m^{3/4})$ \emph{worst case time}. Moreover, Invariants~\ref{inv:neighbors} and~\ref{inv:2hop-neighbors}
	hold after this update.
\end{lemma}

\subsection{The Update Algorithm}\label{sec:dynamic-m-update}

The update algorithm is applied following edge insertions and deletions to and from the graph. After
any edge update, the algorithm updates the data structure and $\mis$. In order to do the latter task, the algorithm may need to remove and/or insert multiple vertices from and to $\mis$.
Since we already argued that maintaining the data structure requires $O(m^{3/4})$ amortized time (by Lemma~\ref{lem:ds}), from now on, without loss of generality, we only
measure the time needed to fix $\mis$ after any edge update and ignore the additive term needed to update the data structure.
The following is the core invariant that we aim to maintain in our algorithm.
\begin{invariant}[Core Invariant]\label{inv:core}
	Following every edge update, the set $\mis$ maintained by the algorithm is an MIS of the input graph. Moreover,
	\begin{enumerate}[label=(\roman*),itemsep=0pt]
	\item\label{part:core1} if only a single vertex leaves $\mis$, then there is no restriction on the number of vertices joining $\mis$ (which could be zero).
	\item\label{part:core2} if at least two vertices leave $\mis$, then at least twice as many vertices join $\mis$.
	\end{enumerate}
	In either case, the total time spent by the algorithm to fix $\mis$ for an edge update is at most an $O(m^{3/4})$ factor larger than the total number of vertices leaving and joining $\mis$.
\end{invariant}
Before showing how to maintain Invariant~\ref{inv:core}, we present the proof of Lemma~\ref{lem:dynamic-m} using this invariant.
\begin{proof}
The main idea behind the proof is as follows.
By Invariant~\ref{inv:core}, after each step, the size of $\mis$ either decreases by at most one, or it will increase.
At the same time, $\mis$ cannot grow more than $n$, the number of vertices in the graph. It then follows that the \emph{average} number of
changes to $\mis$ per each update is $O(1)$. As we only spend $O(m^{3/4})$ per each update,
we obtain the final result. We now present the formal proof using the following charging scheme.

Recall that we compute an MIS $\mis$ of the initial graph in the preprocessing step and that the initialization phase takes $O(m \cdot m^{3/4})$ time in total.
We place $O(m^{3/4})$ ``extra budget'' on vertices in the initial graph that do not belong to $\mis$ to be spent
later when these vertices are inserted to $\mis$. As the number of  such vertices is $O(m)$, this
extra budget can be charged to the time spent in the initialization phase. Note that at this point, an extra budget is allocated to any vertex not in $\mis$ and we maintain this throughout the algorithm.

Whenever an update results in only a single vertex leaving $\mis$ (corresponding to Part \ref{part:core1} of Invariant~\ref{inv:core}),
we spend $O(m^{3/4})$ time to handle this vertex and additionally place $O(m^{3/4})$ budget on this vertex and then for the vertices inserted to $\mis$, we simply use the extra budget allocated to these vertices before
to charge for the $O(m^{3/4})$ time needed to handle each.
If an update results in removing $k > 1$ vertices from $\mis$, we know that at least $2 \cdot k$ vertices would be added to $\mis$ after this update (corresponding to Part \ref{part:core2} of Invariant~\ref{inv:core}).
In this case, we use the $O(m^{3/4})$ extra budget on these (at least) $2\cdot k$ vertices that are joining $\mis$ to charge for the time needed to insert these vertices to $\mis$, remove the $k$ initial vertices from $\mis$, and place $O(m^{3/4})$ extra budget
on every removed vertex. As a result, this type of updates can be handled free of charge. Finally, if an update only involves inserting some vertices to $\mis$, we simply use the budgets on these vertices to handle them free of charge.
This finalizes the proof of Lemma~\ref{lem:dynamic-m}.

We point out that using the above charging scheme, we can also argue that the average number of changes to $\mis$ is $O(1)$ in each update.
\end{proof}

Fix a time step $t$ and suppose the invariant holds up until this time step. Let $e_t = (u_t,v_t)$ be the edge updated at this time step.
In the remainder of this section, we describe one round of the update algorithm to handle this single edge update and preserve Invariant~\ref{inv:core}.

\subsubsection{Edge Deletions}\label{sec:deletions}

We start with the easier case of deleting an edge $e_t = (u_t,v_t)$.

\smallskip
\noindent
\emph{\textbf{Case 1:} Neither $u_t$ nor $v_t$ belong to $\mis$.} In this case, there is nothing to do.

\smallskip
\noindent
\emph{\textbf{Case 2:} $u_t$ belongs to $\mis$ but not $v_t$ (or vice versa).} After deleting the edge $e_t$, it is possible that $v_t$ may need to join $\mis$ as well. We first check whether $\MISNei{v_t} = 0$.
If not, there is nothing else to do as $v_t$ is still adjacent to some vertex in $\mis$. Otherwise, we need to ensure that $v_t$ does not have any neighbor in $\mis$ (outside those vertices counted in $\MISNei{v_t}$).
If $v_t \in V \setminus \vl$, by Invariant~\ref{inv:neighbors},
$\MISNei{v_t}$ counts \emph{all} neighbors of $v_t$ and hence there is nothing more to check. If $v_t \in \vl$, we can go over all the $O(m^{1/4})$ vertices in the neighborhood of $v_t$ and check
whether $v_t$ has a neighbor in $\mis$ or not. This only takes $O(m^{1/4})$ time in the worst case. Again, if we find a neighbor in $\mis$ there is nothing else to do. Otherwise, we add $v_t$ to $\mis$ and update
the data structure which takes $O(m^{3/4})$ time in the worst case by Lemma~\ref{lem:mis-update}. After this step, $\mis$ is again a valid MIS and hence Invariant~\ref{inv:core} is preserved as
we only spent $O(m^{3/4})$ time and inserted at most one vertex to $\mis$ without deleting any vertex from it.


\smallskip
\noindent
\emph{\textbf{Case 3:} Both $u_t$ and $v_t$ belong to $\mis$.} This case is not possible in the first place by Invariant~\ref{inv:core} as otherwise $\mis$ maintained by the algorithm before this edge update was not an MIS.

\subsubsection{Edge Insertions}\label{sec:insertions}

We now consider the by far more challenging case of edge insertions where we concentrate bulk of our efforts.
It is immediate to see that the only time we need to handle an edge insertion is when the inserted edge $e_t = (u_t,v_t)$ connects two vertices already in $\mis$ (there is nothing to do in the remaining cases).
Hence, in the following, \emph{we assume both $u_t$ and $v_t$ belong $\mis$}.

To ensure that $\mis$ is an independent set, we first need to remove one of $u_t$ or $v_t$ from it and then potentially insert some of the neighbors of the deleted vertex to $\mis$ to ensure its maximality.
Let $u_t$ be the deleted vertex (the choice of which vertex to delete is arbitrary). After deleting $u_t$, we update the algorithm's data structure in $O(m^{3/4})$ time by Lemma~\ref{lem:mis-update}.

Let $L := \Nei{u_t} \cap \vl$ denote the set of low degree neighbors of $u_t$.
We first show that one can easily handle all neighbors of $u_t$ which are \emph{not} in $L$. To do so, we can iterate over these vertices as there are $O(m^{3/4})$ of them and for any vertex $w$, by Invariant~\ref{inv:neighbors}, we know
whether $w$ can be added $\mis$ or not by simply checking $\MISNei{w}$. Hence, we can add the necessary vertices to $\mis$ and spend $O(m^{3/4})$ time for each inserted one using Lemma~\ref{lem:mis-update}.
As such, we spend $O(m^{3/4})$ time for iterating the vertices which did not join $\mis$ and $k \cdot O(m^{3/4})$ time for the $k$ vertices that joined $\mis$.
Hence, Invariant~\ref{inv:core} is preserved after this step.

We now consider the challenging case of updating the neighbors of $u_t$ that belong to $L$. As the number of such vertices is potentially very large, we cannot iterate over all of them anymore.
Define the following subsets of $L$:
\begin{itemize}[leftmargin=*,itemsep=0pt]
	\item $\Lmis \subseteq L$: the set of vertices in $L$ that do not have any neighbor in $\mis$.  
	\item $\Lonehop \supseteq \Lmis$: all vertices $w \in L$ where $\MISNei{w} = 0$ i.e., our algorithm did not count any  neighbor for them in $\mis$. Recall that $\MISNei{w}$ does \emph{not} count
	all neighbors of $w$ in $\mis$; it is missing the vertices in $\vh$ by Invariant~\ref{inv:neighbors}. 
	\item $\Ltwohop \supseteq \Lonehop \supseteq \Lmis$: all vertices $w \in L$, where $\MIStwoHopNei{u_t}[w] = 0$. Again, recall that $\MIStwoHopNei{u_t}[w_t]$ does \emph{not} count
	all neighbors of $w \in \MISNei{w}$ (and consequently in $\mis$); it misses the vertices in $\vmh$ in $\MISNei{w}$ (and additionally $\vh$ in $\mis$)
	by Invariant~\ref{inv:2hop-neighbors}. 
\end{itemize}

Let $\ellmis := \card{\Lmis}$, $\ellonehop := \card{\Lonehop}$ and $\elltwohop := \card{\Ltwohop}$, where $\ellmis \leq \ellonehop \leq \elltwohop$.
Our algorithm does not know the sets $\Lmis$ and $\Lonehop$ or even their sizes. However, the update algorithm knows the value of $\elltwohop$ and
has access to vertices in $\Ltwohop$ through the list $\MIStwoHopNei{u_t}$ and can iterate over them in $O(1)$ time per each vertex in $\Ltwohop$ (notice that even this can be potentially too time
consuming as size of this list can be too large). We consider different cases based on the value of these parameters.

\smallskip
\noindent
\emph{\textbf{Case 1:} when $\elltwohop$ is small, i.e., $\elltwohop \leq 4 \cdot m^{3/4}$.}
In this case, we iterate over vertices $w \in \Ltwohop$ in $O(\elltwohop) = O(m^{3/4})$ time and check whether $\MISNei{w} = 0$ or not.
This allows us to compute the set $\Lonehop$ and $\ellonehop$ as well. We further distinguish between two cases.

\smallskip
\noindent
\emph{\textbf{Case 1-a:} when $\ellonehop$ is very small, i.e., $\ellonehop \leq 4 \cdot m^{1/2}$.}
We iterate over vertices $w \in \Lonehop$ and for each vertex, spend  $O(m^{1/4})$  time to
go over all its neighbors and decide whether $w$ has any neighbor in $\mis$ or not (degree of $w$ is $O(m^{1/4})$ since it belongs to $\vl$). Hence, in this case, we can obtain the
set $\Lmis$ fully in $O(m^{3/4})$ time in total.

We then iterate over vertices in $\Lmis$, insert each one greedily to $\mis$, and update the data structure in $O(m^{3/4})$ time using Lemma~\ref{lem:mis-update}. It is possible that some vertices in $\Lmis$ are adjacent to each other and hence
before inserting any vertex $w$, we first need to check $\MISNei{w}$ to make sure it is zero still (by Invariant~\ref{inv:neighbors} and since all vertices in $\Lmis$ belong to $\vl$, any vertex added to $\mis$ here would update
$\MISNei{w}$ for any neighbor $w$). Hence, in this case, we spend $O(m^{3/4})$ time for each vertex inserted to $\mis$ and did not delete any vertex from it. Therefore, Invariant~\ref{inv:core} is preserved after the edge update
in this case.

\smallskip
\noindent
\emph{\textbf{Case 1-b:} when $\ellonehop$ is \emph{not} very small, i.e., $\ellonehop > 4 \cdot m^{1/2}$.}
In this case, we cannot afford to compute $\Lmis$ explicitly. Rather, we simply add the vertices in $\Lonehop$ to $\mis$ directly,
without considering whether they are adjacent to vertices already in $\mis$ or not at all (although we check that they are not adjacent to the previously inserted vertices from $\Lonehop$).
As a result, it is possible that after this process, $\mis$ is not an independent set of the graph anymore. To fix this,
we perform a post processing step in which we delete some vertices from $\mis$ to ensure that the remaining vertices indeed form an MIS of the original graph.

Concretely, we go over vertices in $\Lonehop$ and insert each to $\mis$ if none of its neighbors have been added to $\mis$ \emph{in this step}, and then invoke Lemma~\ref{lem:mis-update} to update the algorithm's data structure. Since in this step, we
are only adding vertices that are in $\vl$, we can check in $O(1)$ time whether a vertex has a neighbor in $\mis$ (that has been added in this step) or not by Invariant~\ref{inv:neighbors}. This step clearly takes $O(m^{3/4})$ time per each vertex inserted to the
MIS.

At this point, it is possible that there are some vertices in $\mis$ which are adjacent to the newly inserted vertices. By Invariant~\ref{inv:neighbors}, we know that these vertices can only belong to $\vh$ and hence there are at most $m^{1/4}$
of them. We iterate over all vertices in $\vh$ and check whether they have a neighbor in $\mis$ (by Invariant~\ref{inv:neighbors}, we stored this information for these vertices) and mark all such vertices. Next,
we remove all these marked vertices from $\mis$ simultaneously and update the algorithm's state by Lemma~\ref{lem:mis-update}. We are not done yet though because after removing these vertices, it is possible that we may need to bring
some of their neighbors back to $\mis$. We solve this problem recursively using the same update algorithm by treating these marked vertices the same as $u_t$.

We argue that Invariant~\ref{inv:core} is preserved. As the degree of vertices in $\Lonehop$ is bounded by $m^{1/4}$, the number of vertices
added to $\mis$ in this part is at least $\frac{\ellonehop}{(m^{1/4}+1)} \geq 2 \cdot m^{1/4}$ (by Fact~\ref{fact:greedy-mis} and the assumption on $\ellonehop$ in this case). On the other hand, the number of vertices removed from $\mis$
is at most equal to size of $\vh$ which is $m^{1/4}$. As a result, in this specific step, the number of vertices inserted to $\mis$ is at least twice as many as the vertices removed from it.
For any vertex inserted or deleted from $\mis$ also, we spent $O(m^{3/4})$ time. As we are performing the recursive step using the same algorithm, we can argue inductively that for any vertex deleted in those recursive calls, at least
twice as many vertices would be added to $\mis$ and that the total running time would be proportional to the number of vertices added or removed from $\mis$ times $O(m^{3/4})$. We point out that any recursive call that leads to another one in this algorithm
necessarily increase the number of vertices in $\mis$ and hence the algorithm does indeed terminate (see also case \emph{2}).

\smallskip
\noindent
\emph{\textbf{Case 2:} when $\elltwohop$ is \emph{not} small, i.e., $\elltwohop > 4 \cdot m^{3/4}$.}
We use a similar strategy as case $\emph{1-b}$ here as well. We iterate over all vertices in $\Ltwohop$, greedily add each vertex to $\mis$  as long as this vertex is not adjacent to any of the newly added vertices (which can be checked in $O(1)$ time by Invariant~\ref{inv:neighbors}), and
update the data structure using Lemma~\ref{lem:mis-update}. As the maximum degree of
vertices in $\Ltwohop$ is at most $m^{1/4}$, we add at least $\frac{\elltwohop}{m^{1/4} + 1} > 2 \cdot m^{1/2}$ vertices to $\mis$ by Fact~\ref{fact:greedy-mis}. By Invariant~\ref{inv:2hop-neighbors}, if a vertex belongs to $\Ltwohop$, the only neighbors
of this vertex in $\mis$ belong to $\vh$ or $\vmh$ and hence has degree at least $m^{1/2}$. We go over these vertices next and mark them. Then, we remove all of them from $\mis$ simultaneously and update the algorithm by Lemma~\ref{lem:mis-update}.
Similar to case \emph{1-b}, we now also have to consider bringing some of the neighbors of these vertices to $\mis$ which is handled recursively exactly the same way as in case \emph{1-b}.

We first analyze the time complexity of this step. Iterating over $\Ltwohop$ takes $O(\card{\elltwohop}) = O(m)$ time and since we are inserting at least $m^{1/2}$ vertices from $\Ltwohop$ to $\mis$, we can charge the time needed
for this step to the time allowed for inserting these vertices to $\mis$. Moreover, we inserted at least $2 \cdot m^{1/2}$ vertices to $\mis$ and would remove at most $m^{1/2}$ vertices after considering violating vertices in $\vh$ and $\vmh$. Hence,
number of inserted vertices is at least twice the number of removed ones at this step. We can also argue inductively that this property hold for each recursive call similar to the case \emph{1-b}. This finalizes the proof of this case.

\medskip

To conclude, we proved that Invariant~\ref{inv:core} is preserved after any edge insertion or deletion in the algorithm, which finalizes the proof of Lemma~\ref{lem:dynamic-m}.

\section{Maximal Independent Set in Dynamic Distributed Networks}\label{app:distributed}

We consider the $\CONGEST$ model of distributed computation (cf.~\cite{PelB00}) which captures the essence of both spatial locality and congestion.
The network is modeled by an undirected graph $G(V,E)$ where the vertex-set is $V$, and $E$ corresponds to both the edge-set in the current graph and also the vertex pairs that can directly communicate with each other.
We assume a synchronous communication model, where time is divided into rounds and in each round, each vertex can send a message of size $O(\log{n})$ bits to any of its neighbors, where $n = \card{V}$.
The goal is to maintain an MIS $\mis$ in $G$ in a way that each vertex is able to output whether or not it belongs to $\mis$.

We focus on dynamically changing networks where both edges and vertices can be inserted to or deleted from the network. For deletions, we consider
\emph{graceful deletions} where the deleted vertex/edge may be used for passing messages between its neighbors (endpoints),
and is only deleted completely once the network is stable again.
After each change, the vertices communicate with each other to adjust their outputs, namely make the network \emph{stable} again.
We make the standard assumption that the changes occur in large enough time gaps, and hence the network is always stable before the next change occurs (see, e.g.,~\cite{ParterPS16,CHK16}).
We further assume that each change in the network is indexed and vertices affected by this change know how many updates have happened before\footnote{This is only needed by our algorithm in Theorem~\ref{thm:distributed-m} to
have an approximation of the number of edges in the graph, which is a global quantity and cannot be maintained by each vertex locally}.

There are three complexity measures for the algorithms in this model. The first is the so-called \emph{adjustment complexity}, which measures the number of vertices that change their output as a result
of a recent topology change. The second is the \emph{round complexity}, the number of rounds required for the network to become stable again after each update.
The third is the \emph{message complexity}, measuring the total number of $O(\log{n})$-length messages communicated by the algorithm.

Our main result in this section is an implementation of Theorem~\ref{thm:dynamic-m} in this distributed setting for maintaining an MIS in a dynamically changing network.
\begin{theorem}\label{thm:distributed-m}
	Starting from an empty distributed network on $n$ vertices, a maximal independent set can be maintained \emph{deterministically} in a distributed fashion (under the \CONGEST communication model)
	over any sequence of vertex/edge insertions and (graceful) deletions with $(i)$ $O(1)$ \emph{amortized adjustment complexity}, $(ii)$ $O(1)$ \emph{amortized round complexity}, and $(iii)$ $O(m^{3/4})$ \emph{amortized message complexity}.
	Here, $m$ denotes the number of dynamic edges.
\end{theorem}

The algorithm in Lemma~\ref{lem:dynamic-delta} can also be trivially implemented in this distributed setting, resulting in \emph{an extremely simple deterministic distributed algorithm for maintaining an MIS of a dynamically
changing graph in $O(1)$ amortized adjustment complexity and round complexity, and $O(\Delta)$ amortized message complexity}. As argued before, this simple algorithm already strengthens the previous randomized
algorithm of Censor-Hillel~\etal~\cite{CHK16} by virtue of being deterministic and not requiring an assumption of a non-adaptive oblivious adversary.  In the following, we compare our results in the distributed setting with
those of~\cite{CHK16}.

\paragraph{Amortized vs in Expectation Guarantee.} The guarantees on the complexity measures provided by our deterministic algorithms in this setting are \emph{amortized}, while the randomized algorithm in~\cite{CHK16} achieves
its bound \emph{in expectation} which may be considered somewhat stronger than our guarantee. To achieve this guarantee however, the algorithm in~\cite{CHK16}, besides using randomization, also assumes
a non-adaptive oblivious adversary. An adaptive adversary (the assumption supported by all our algorithms in this paper) can force the algorithm in~\cite{CHK16} to adjust the MIS by $\Omega(n)$ vertices in every round, which
in turn blows up all the complexity measures in~\cite{CHK16} by a factor of $\Omega(n)$. It is also worth mentioning that the guarantee achieved by~\cite{CHK16} only holds in expectation and not \emph{with high probability} and for a
fundamental reason: It was shown in~\cite{CHK16} that for every value of $k$, there exists an instance for which at least $\Omega(k)$ adjustments are needed for any algorithm with probability at least $1/k$ (see Section 1.1 of their paper).

\paragraph{Broadcast vs Unicast.} The communication in algorithm of~\cite{CHK16} in each round is $O(1)$ broadcast messages in expectation that requires only $O(1)$ bits on every edge
(i.e., each vertex communicates the same $O(1)$ bits to every one of its neighbors). As such, the total communication at every round of this algorithm is $O(\Delta)$ bits in expectation.
Our amortized $O(\Delta)$-message complexity algorithm (distributed implementation of Lemma~\ref{lem:dynamic-delta}) also works with the same guarantee: indeed, every
vertex simply needs to send $O(1)$ bits to all its neighbors in a broadcast manner so that their neighbors know whether to add or subtract the contribution of this vertex to or from their counter.
This is however not the case for our main algorithm in Theorem~\ref{thm:distributed-m} which requires a processor to communicate differently to its neighbor over each edge (in general, one cannot hope to achieve $o(\Delta)$ communication
with only broadcast messages). Additionally, this algorithm
now requires to communicate $O(\log{n})$ bits (as opposed to $O(1)$ in the previous two algorithms) over every edge. This is mainly due to the fact that in this new algorithm we need to communicate
with vertices which are at distance $2$ of the current vertex and hence we need to carry the ID of original senders in the messages also.

\paragraph{Graceful vs Abrupt Deletions.} A stronger notion of deletion in the dynamic setting is \emph{abrupt deletion} in which the neighbors of the deleted vertex/edge simply discover that this vertex/edge is being deleted and
the deleted vertex/edge cannot be used for communication anymore right after the deletion happens. Censor-Hillel~\etal~\cite{CHK16} also extend their result to this more general setting and achieved the same guarantees except for
message complexity of abrupt deletion of a node which is now $O(\min\set{\log{n},\Delta})$ broadcasts as opposed to $O(1)$. We do not consider this model explicitly. However, it is straightforward to verify that our
amortized $O(\Delta)$-message complexity algorithm (distributed implementation of Lemma~\ref{lem:dynamic-delta}) works in this more general setting with virtually no change and even still achieves amortized
$O(1)$ broadcast per abrupt deletion of a vertex as well. We believe that our main algorithm in Theorem~\ref{thm:distributed-m} should also work in this more general setting with proper modifications but we did not prove this formally.

\paragraph{Synchronous vs Asynchronous Communication.} We focused only on the synchronous communication in this paper. Censor-Hillel~\cite{CHK16} also considered the asynchronous
model of communication and showed that their algorithm holds in this model as well, albeit with a weaker guarantee on its message complexity. Our algorithms can be modified to work in an asynchronous model as well, as at each stage of the algorithm
we can identify a (different) local ``coordinator'' that can be used to synchronize the operations with an added overhead that is within a constant multiplicative of the synchronous complexity (as per each update only vertices within two-hop neighborhood of a vertex
need to communicate with each other in our algorithm); we omit the details but refer the reader to Section~\ref{sec:distributed-m-edge-insertions} for more information on the use of a local coordinator in our algorithms.

\medskip

We now turn to proving Theorem~\ref{thm:distributed-m}, using the following lemma the same way we used Lemma~\ref{lem:dynamic-m} in the proof of Theorem~\ref{thm:dynamic-m}.

\begin{lemma}\label{lem:distributed-m}
	Starting with any arbitrary graph on $n$ vertices and $m$ edges, a maximal independent set $\mis$ can be maintained \emph{deterministically} in a distributed fashion (under the \CONGEST communication model) over any
	sequence of $K = \Omega(m)$ vertex/edge insertions and (graceful) deletions as long as the number of edges in the graph remains within a factor $2$ of $m$. The algorithm:
	\begin{enumerate}[label=(\roman*)]
		\item makes $O(K)$ adjustment to $\mis$ in total, i.e., has $O(1)$ amortized adjustment complexity,
		\item requires $O(K)$ rounds in total, i.e., has $O(1)$ amortized round complexity, and
		\item communicates $O(K \cdot m^{3/4})$ messages in total, i.e., has $O(m^{3/4})$ amortized message complexity.
	\end{enumerate}
\end{lemma}

The algorithm in Lemma~\ref{lem:distributed-m} is a simple implementation of our sequential dynamic algorithm in Lemma~\ref{lem:dynamic-m}. In the following, we first adapt the data structures
introduced in Section~\ref{sec:dynamic-m-ds} to the distributed setting. We then show that with proper adjustments,  the (sequential) update algorithm in Section~\ref{sec:dynamic-m-update} can also be used
in the \CONGEST model and prove Theorem~\ref{thm:distributed-m}.

\subsection{The Data Structure}\label{sec:distributed-m-ds}

We store the same exact information in Section~\ref{sec:dynamic-m-ds} per each vertex here as well and maintain Invariants~\ref{inv:neighbors} and~\ref{inv:2hop-neighbors}.
We first prove that the two
procedures $\UpdateNeighbors$ and $\UpdateTwoHopNeighbors$ can both be implemented in constant rounds and $O(m^{3/4})$ messages in total. In particular,

\begin{lemma}\label{lem:distributed-m-updates}
	For any vertex $u \in V$,
	\begin{enumerate}[label=(\roman*)]
	\item $\UpdateNeighbors(u)$ operation requires spending $1$ round and $m^{3/4}$ messages in total, and
	\item $\UpdateTwoHopNeighbors(u)$ operation requires $2$ rounds and $2\cdot m^{3/4}$ messages in total.
	\end{enumerate}
\end{lemma}
\begin{proof}
	\emph{Part $(i)$.} If $u \in \vh$, it only needs to send a message to its neighbors in $V \setminus \vl$ and inform them on the status of $u$ (whether it is inserted to or deleted from $\mis$), which requires
	only $1$ round (as they are all neighbors to $u$) and $m^{3/4}$ messages as $\card{V \setminus \vl} \leq m/m^{1/4} = m^{3/4}$. If $u \notin \vh$, it would update all its neighbors again in $1$ round
	and $m^{3/4}$ messages as the latter is an upper bound on number of its neighbors.

	\emph{Part $(ii)$.} If $u \notin \vl \cup \vml$ there is nothing to do. Otherwise, $u$ needs to send a message to all its (at most $m^{1/2}$)
	neighbors that belong to $\vl$ and ask them to relay this information
	to their neighbors. These vertices can then spend another round to inform all  their (at most $m^{1/4}$) neighbors about the status of $u$. This takes $2$ rounds and $m^{1/2} + m^{1/2} \cdot m^{1/4} < 2\cdot m^{3/4}$ messages in total.
\end{proof}

Lemma~\ref{lem:distributed-m-updates} ensures that Invariants~\ref{inv:neighbors} and~\ref{inv:2hop-neighbors} (the only MIS-related information stored for $u$ beside $\MISFlag{u}$ that can be trivially updated) are preserved after
any change in $\mis$ within a constant number of rounds and $O(m^{3/4})$ messages.

In the following, we briefly describe how to update the information stored for vertices per each topology change in the graph.

\paragraph{Vertex Updates.} Let $u$ be the updated vertex. In case of vertex insertion, we simply initialize the data structures at $u$ and we are almost done as the neighbors of $u$ are already informed about $u$ being inserted to the graph
and hence can update their information locally. We only need to send $\Deg{u}$ to all the neighbors (the time needed for this can be charged to the initialization cost of this algorithm).
Now suppose $u$ is being deleted. We can update $\Nei{v}$ and $\NeiDeg{v}$ for any neighbor $v$ of $u$ without any communication as they are informed that $u$ is deleted. We can also run $\UpdateNeighbors(u)$ (virtually) with no communication
as this procedure only informs the neighbors of $u$ that this vertex is being deleted from $\mis$ and by knowing that $u$ has left the graph, any vertex $v$ in the neighborhood of $u$ can update $\MISNei{v}$ accordingly. Finally, we can also
run $\UpdateTwoHopNeighbors(u)$ with only $1$ round of communication and $m^{3/4}$ messages (see Lemma~\ref{lem:distributed-m-updates})
by relaying the information from the neighbors of $u$ (which are informed about $u$ leaving the graph) to their neighbors.

\paragraph{Edge Updates.} These updates are handled exactly the same as in our sequential algorithm. Let $(u,v)$ be the updated edge. Nertices $u$ and $v$ can update all information
except for updating $\MIStwoHopNei{\cdot}$ (in particular $\NeiDeg{\cdot}$ can be updated by the procedure described in Section~\ref{sec:dynamic-m-ds} with $O(1)$ worst-case round complexity and $O(1)$ amortized message complexity).
To do the latter task, vertex $u$ (resp. $v$) can simulate $\UpdateTwoHopNeighbors(v)$ (resp. $\UpdateTwoHopNeighbors(u)$) as described above, which takes
$m^{3/4}$ messages and $1$ round.

We hence showed that after each change in the topology, all the information stored for vertices can be updated in $O(1)$ rounds and $O(m^{3/4})$ amortized messages.

\subsection{The Distributed Algorithm}\label{sec:distributed-m-alg}

We design a distributed algorithm for updating $\mis$ in the network in the spirit of our update algorithm in Section~\ref{sec:dynamic-m-update}.
The algorithm is a simple adaption of our sequential algorithm to this dynamic model.
For every update, we first perform the steps in the previous section to update the information on every vertex in the graph and then make the network stable again by adjusting $\mis$.

Throughout, we aim to maintain the following invariant which is the direct analogue of Invariant~\ref{inv:core} in the dynamic setting.

\begin{invariant}\label{inv:distributed}
	Following every vertex/edge update, the set $\mis$ maintained by the algorithm is an MIS of the input graph. Moreover,
	\begin{enumerate}[label=(\roman*)]
	\item if only a single vertex leaves $\mis$ then there is no restriction on the number of vertices joining $\mis$ (which could be zero).
	\item if at least two vertices leave $\mis$, then at least twice as many vertices are added to $\mis$.
	\end{enumerate}
	In either case, the worst case number of rounds and messages spent by the algorithm for any update is within, respectively, an $O(1)$ and an $O(m^{3/4})$ factor of the total number of vertices leaving and joining $\mis$.
\end{invariant}

Using the same exact argument as in the proof of Lemma~\ref{lem:dynamic-m}, maintaining Invariant~\ref{inv:distributed} ensures that the amortized adjustment complexity and amortized round complexity of the algorithm is $O(1)$
and its amortized message complexity is $O(m^{3/4})$. Hence, to prove Lemma~\ref{lem:distributed-m}, it suffices to prove that Invariant~\ref{inv:distributed} is preserved after every update.
We consider different cases based on insertion and deletion of edges and vertices.

\subsubsection{Edge Deletions}\label{sec:distributed-m-edge-deletions}

Suppose we delete the edge $(u,v)$. We only consider the case that $u$ belongs to $\mis$ and $v$ is not; the remaining cases are either symmetric to this one or need no update in $\mis$ (see Section~\ref{sec:deletions}).
If $v$ is not in $\vl$, by Invariant~\ref{inv:neighbors}, it knows all its neighbors in $\mis$ and can decide whether to join $\mis$ or not to locally; if it enters $\mis$, it can update the network in $O(1)$ rounds and
$O(m^{3/4})$ messages by Lemma~\ref{lem:distributed-m-updates}. If $v$ is in $\vl$, it first sends a message to all its $O(m^{1/4})$ neighbors and ask for their status to which its neighbors reply whether they belong to $\mis$ or not.
This only takes $2$ rounds and $O(m^{1/4})$ communication and then $v$ can decide again whether to join $\mis$ or not to. Note that this part of the result holds even with abrupt deletions.

\subsubsection{Edge Insertions}\label{sec:distributed-m-edge-insertions}

Suppose we insert the edge $(u,v)$. We only consider the case when both $u$ and $v$ belong to $\mis$; the remaining cases need no update in $\mis$ (see Section~\ref{sec:insertions}). Remember that in Section~\ref{sec:insertions}, we needed
to handle these updates in three separate cases. While the algorithm and analysis in each case is different, the procedures needed to carry the information around the network are essentially the same among these cases and hence in the following, for simplicity,
we only consider one of the main cases, namely case \emph{1-b} (see Section~\ref{sec:insertions} for definition of this case). The algorithm in the remaining cases can be adapted to this setting in the same exact way.

Recall that in this case, the vertex $u$ is deleted from $\mis$ and moreover $u$ knows the set $\Lonehop$ entirely, which is of size $O(m^{3/4})$. The general approach is to make $u$ a ``coordinator'' for running the update algorithm
in Section~\ref{sec:insertions} by communicating with its two-hop neighborhood and gather the necessary information to run the sequential update algorithm.

Vertex $u$ first sends a message to all its neighbors in $\Lonehop$
and asks for their status to which they respond whether or not they belong to $\mis$. This takes $2$ rounds and $O(m^{3/4})$ messages.
Next, $u$ informs one of its neighbors that it can join $\mis$ and this new vertex updates its status and the information in the graph which takes $O(1)$
rounds and $O(m^{3/4})$ messages by Lemma~\ref{lem:distributed-m-updates}. After this, $u$ again sends a message to all its neighbors in $\Lonehop$ and asks for their status in $\mis$ to which they respond whether they belong to $\mis$ or whether one of their neighbors in $\Lonehop$ has been added to $\mis$ in this step. Then, again, $u$ informs one of its neighbors (if such exists) that it can join $\mis$, and continues. This way, we only spend $O(1)$ rounds
and $O(m^{3/4})$ communication per each vertex entering $\mis$ in addition to $O(1)$ rounds and $O(m^{3/4})$ communication for communicating with neighbors of $u$ that would not join $\mis$ eventually.

After processing the list $\Lonehop$, we also need to delete from $\mis$, the set of vertices in $\vh$ that are now incident to vertices in $\Lonehop$ that just joined $\mis$. Note that such vertices are necessarily in the two-hop neighborhood
of $u$ and hence $u$ can communicate with them (which are only $O(m^{1/4})$ many) in $O(1)$ rounds and use the above idea to implement the same update algorithm in Section~\ref{sec:insertions} in this model. This allows us to
preserve Invariant~\ref{inv:distributed} by the same exact analysis in Section~\ref{sec:insertions}.

\subsubsection{Vertex Deletions}\label{sec:distributed-m-vertex-deletions}

This case is essentially equivalent to the edge insertion case discussed above. Since we have a graceful deletion, we can treat the deleted vertex the same way as in Section~\ref{sec:distributed-m-edge-insertions} by deleting it from
$\mis$ (if it belonged to it) and using it as the ``coordinator'' to implement the process described in Section~\ref{sec:distributed-m-edge-insertions}.

 \subsubsection{Vertex Insertions}\label{sec:distributed-m-vertex-insertions}

The only thing we need to do in this case is to check whether we need to add this new vertex to $\mis$ or not. If this vertex is not in $\vl$, it already knows this information and hence can decide whether or not to join $\mis$; after
that we are done. Otherwise, if the vertex belongs to $\vl$, it sends a message to all its $O(m^{1/4})$ neighbors and ask for their status in $\mis$, and use that to decide about joining $\mis$. In either case,
we only need $O(1)$ rounds and $O(m^{1/4})$ total communication. After this, we update the neighbors using first part of Lemma~\ref{lem:distributed-m-updates} in $O(m^{3/4})$ communication and $O(1)$ rounds.

\medskip

To conclude, we showed that Invariant~\ref{inv:distributed} is preserved after any edge or vertex insertion or deletion by the distributed algorithm, hence proving Lemma~\ref{lem:distributed-m}. We are now ready to prove
Theorem~\ref{thm:distributed-m}.

\begin{proof}[Proof of Theorem~\ref{thm:distributed-m}]
	The proof is identical to the proof of Theorem~\ref{thm:dynamic-m}. The only difference is that in this distributed setting, we are not able to maintain the exact number of edges in the graph in a distributed fashion across all vertices.
	However, recall that we assumed vertices affected by an update in the topology know the index of this update, i.e., how many updates have happened before this one. Hence, whenever the number of updates reaches $\Omega(m)$,
	any vertex that knows this information sends a message to all its neighbors to terminate the process which would then be broadcast across the whole graph. This takes $O(m)$ rounds and $O(m)$ communication and can be charged
	to the total number of updates, i.e., $\Omega(m)$ in this step. Hence, the vertices can initialize their data structure using the new choice of $m$ and continue the distributed algorithm in Lemma~\ref{lem:distributed-m}.
\end{proof}




\bibliographystyle{abbrv}
\bibliography{randomMMbibfile}

\clearpage
\appendix

\section{An $\Omega(n)$ Lower Bound on Worst-Case Adjustment Complexity}\label{app:worst-case-example}

By a straightforward modification of the lower bound example in~\cite{CHK16}, we can show that adjustment-complexity of any deterministic algorithm is $\Omega(n)$ in the worst-case.

Consider the following example: Let $G_1$ be a complete bipartite graph between two sets of vertices
$L_1$ and $R_1$, each of size $n/4$. We create an identical copy of $G_1$ named $G_2$ with bipartition $L_2,R_2$ on the remaining vertices and let $G$ be the union of these graphs.
Consider any deterministic algorithm $\alg$ for maintaining an MIS on $G$.
Without loss of generality, assume $L_1 \cup L_2$ is the MIS chosen by $\alg$ (in any MIS of $G$ either $L_1$ is entirely in the MIS or $R_1$ and similarly for $L_2$ and $R_2$).
Let $u_1$ and $u_2$ be two arbitrary vertices in $L_1$ and $L_2$. The adversary starts deleting all edges incident to all vertices in $L_1 \setminus \set{u_1}$ and $L_2 \setminus \set{u_2}$.  Finally, it adds an edge between $u_1$ and $u_2$.
We argue that at some point during these updates, $\alg$ adjusted $\Omega(n)$ vertices in the maintained MIS.
There are two cases to consider. For simplicity of exposition, we assume that at each time step, \emph{all} edges incident to a vertex are deleted at once.

Suppose at some point before inserting the last edge,
$\alg$ decides to add a vertex in $R_1$ to the MIS for the first time (the argument is symmetric for $R_2$ as well). Since $u_1$ is incident to all vertices in $R_1$, this means that $u_1$ needs to leave the MIS. Also, since a vertex in
$R_1$ has joined the MIS, we know that there cannot be any edge from vertices in the MIS to \emph{any} vertex in $R_1$ (as we start with a complete bipartite graph and assumed that all edges incident to a vertex are deleted simultaneously). This
means that after this step, \emph{all} vertices in $R_1$ should join the MIS to ensure maximality. Therefore, at this step, $\Omega(n)$ vertices are inserted to the MIS at once, proving the claim in this case.

Now suppose that before inserting the last edge, no vertex in $R_1$ and $R_2$ belong to the MIS and hence both $u_1$ and $u_2$ should be inside it. By adding an edge between $u_1$ and $u_2$, $\alg$ is forced to remove
at least one of them, say $u_1$, from the MIS, which in turn forces all $R_1$ to join MIS to keep the maximality. Hence, again, $\Omega(n)$ vertices are inserted to the MIS at once, finalizing the proof.

\begin{remark}
We remark that this simple example explains why we obtain our results in \emph{amortized} bounds rather than \emph{worst-case} bounds. 
\end{remark}

\end{document}